\documentclass[oneeqnum,final]{siamltex1213}

\usepackage{subcaption}
\usepackage[format=hang]{caption}
\usepackage{amsmath}
\usepackage{amsfonts}

\usepackage[utf8]{inputenc} 

\newcommand{\Real}{{\mathbb R}}

\newcommand{\TwoEC}{\emph{2EC}}
\newcommand{\TwoECLP}{\ensuremath{\text{\TwoEC}^{\text{LP}}}}
\newcommand{\alphaTwoEC}{\ensuremath{\alpha\text{\TwoEC}}}

\newcommand{\Kn}{\ensuremath{K_{n}}}
\newcommand{\Gx}{\ensuremath{G_{x^{*}}}}

\newtheorem{conjecture}{Conjecture}
\newtheorem{case}{Case}

\title{Toward a 6/5 Bound for the Minimum Cost 2-Edge Connected Spanning Subgraph Problem \thanks{This research was partially supported by grants from the Natural Sciences and \mbox{Engineering} Research Council of Canada}}

\author{Sylvia Boyd \thanks{Email: \href{mailto:sylvia@site.uottawa.ca}   
   {\texttt{sylvia@site.uottawa.ca}}}
\and Philippe Legault \thanks{Email: \href{mailto:philippe@legault.cc}  
   {\texttt{philippe@legault.cc}}}}

\begin{document}

\maketitle

\begin{abstract}
Given a complete graph $\Kn=(V, E)$ with non-negative edge costs $c\in \Real^{E}$, the problem \TwoEC{} is that of finding a 2-edge connected spanning multi-subgraph of \Kn{} of minimum cost. The integrality gap \alphaTwoEC{} of the linear programming relaxation \TwoECLP{} for \TwoEC{} has been conjectured to be $\frac{6}{5}$, although currently we only know that $\frac{6}{5}\leq\alphaTwoEC\leq\frac{3}{2}$. In this paper, we explore the idea of using the structure of solutions for \TwoECLP{} and the concept of convex combination to obtain improved bounds for \alphaTwoEC. We focus our efforts on a family $J$ of half-integer solutions that appear to give the largest integrality gap for \TwoECLP. We successfully show that the conjecture $\alphaTwoEC = \frac{6}{5}$ is true for any cost functions optimized by some $x^{*}\in J$.
\end{abstract}

\begin{keywords}
minimum cost 2-edge connected subgraph problem, approximation algorithm, integrality gap.
\end{keywords}

\section{Introduction}
The \emph{2-edge connected subgraph problem} (\TwoEC{}) is that of finding a minimum cost 2-edge connected spanning multi-subgraph of the complete graph \sloppy{$\Kn=(V, E)$} with costs $c\in \Real^{E}_{\geq 0}$. This problem has many important applications in network design. It is known to be NP-hard even for very special cases\,\cite{csaba}. Currently, a $\frac{3}{2}$-approximation algorithm is known for \TwoEC{}. This follows from the fact that for any instance of \TwoEC{}, we can assume, WLOG, that the costs are metric and the solutions do not include multi-edges\,\cite{alexander}, in which case we can apply the $\frac{3}{2}$-approximation due to Frederickson and Ja'Ja'\,\cite{frederickson}. For \TwoEC{} where multi-edges are not allowed, a 2-approximation is known \cite{jain}.

For $e\in E$, letting $x_{e}$ represent the number of copies of $e$ in the \TwoEC{} solution, \TwoEC{} can be formulated as an integer linear program (ILP) as follows, i.e.:\begin{equation}\label{ilpTwoEC}
     \begin{aligned}
      \text{Minimize}\quad & cx&  \\
      \text{Subject to}\quad &\sum(x_{ij}:i\in S, j\notin S) \geq 2 &\text{for all } \emptyset \subset S \subset V,\\
       & x_e \ge 0\text{, and integer} &\text{for all } e \in E. \\
     \end{aligned}
\end{equation}
\noindent The linear programming (LP) relaxation of \TwoEC{}, denoted by \TwoECLP, is obtained by relaxing the integer requirement in (\ref{ilpTwoEC}). We use OPT(\TwoEC{}) (resp. OPT(\TwoECLP)) to denote the optimal value of \TwoEC{} (resp. \TwoECLP). Also, given any feasible solution $x^{*}$ for \TwoECLP, its \emph{support graph} \Gx{} is defined to be the subgraph of \Kn{} obtained by taking all edges $e\in E$ for which $x^{*}_{e}>0$.

We are interested in the \emph{integrality gap} \alphaTwoEC{} for \TwoECLP, which is the worst case ratio between OPT(\TwoEC) and OPT(\TwoECLP), i.e. \[\alphaTwoEC=\max_{\substack{c\geq 0\\c\neq 0}} \frac{\text{OPT(}\TwoEC\text{)}}{\text{OPT(}\TwoECLP\text{)}}\text{.}\] This gives a measure of the quality of the lower bound provided by \TwoECLP. Moreover, a polynomial-time constructive proof of $\alphaTwoEC=k$ would provide a $k$-approximation algorithm for \alphaTwoEC.

Even though \TwoEC{} has been intensively studied, little is known about \alphaTwoEC, except that ${\frac{6}{5}\le\alphaTwoEC\le\frac{3}{2}}$\,\cite{alexander} in general, and $\frac{8}{7}\leq\alphaTwoEC\leq\frac{4}{3}$ for the unweighted form of the problem in which one is given a graph and all edge costs are 1 (see\,\cite{boyd} and \cite{sebo}). In\,\cite{carr}, Carr and Ravi study \alphaTwoEC{}, and conjecture that $\alphaTwoEC{} = \frac{4}{3}$, however no examples are known for which the integrality gap ratio comes close to $\frac{4}{3}$. In\,\cite{alexander}, Alexander, Boyd and Elliott-Magwood also study \alphaTwoEC{} and make the following stronger conjecture based on their findings:

\begin{conjecture}\cite{alexander}\label{conjecture}
The integrality gap \alphaTwoEC{} for \TwoECLP{} is $\frac{6}{5}$.
\end{conjecture}

To investigate \alphaTwoEC{} further, a natural next step is to study $\alphaTwoEC$ for some interesting class of cost functions. We investigate \alphaTwoEC{} for the set of cost functions optimized at a particular family of feasible solutions for \TwoECLP. A feasible solution $x^{*}$ for \TwoECLP{} is called a \emph{half-integer solution} if $x^{*}_{e}\in\{0, \frac{1}{2}, 1\}$ for all $x^{*}_e\in E$, and it is called \emph{degree-tight} if $\sum_{uv}(x^{*}_{uv}:u\in V)=2$ for all $v\in V$. Finally, a degree-tight half-integer solution is called a \emph{half-triangle solution} if the edges in the support graph \Gx{} corresponding to $x^{*}_{e}=\frac{1}{2}$ (called \emph{half-edges}) form disjoint 3-cycles (called \emph{half-triangles}) joined by paths of edges of value 1 (called \emph{\mbox{1-paths}}).

The half-triangle solutions are of interest for studies of \alphaTwoEC{} as there is \mbox{evidence} that $\frac{\text{OPT}(\TwoEC)}{\text{OPT}(\TwoECLP)}$ is greatest for cost functions optimized at such solutions (see\,\cite{alexander},\,\cite{carr}). For example, the largest such ratio known is asymptotically $\frac{6}{5}$\,\cite{alexander}, and comes from the infinite family of \TwoEC{} problems shown in Figure \ref{65WorseCaseExample}, where the numbers shown are the edge costs, edges $uv$ not shown have cost equal to the minimum cost $uv$ path, and the ``gadget'' pattern is repeated $k$ times. This family is optimized for \TwoECLP{} by the half-triangle solution $x^{*}$ shown in Figure \ref{65WorseCaseExample2}. Also, in a computational study which found \alphaTwoEC{} exactly for all \Kn{} up to $n=10$ and all half-integer solutions up to $n=14$, \alphaTwoEC{} was given by a half-triangle solution for all values of $n$\,\cite{alexander}.

\begin{figure}
	\centering
	\begin{subfigure}[t]{0.45\textwidth}
		\vskip 0pt
		\centering
		\includegraphics{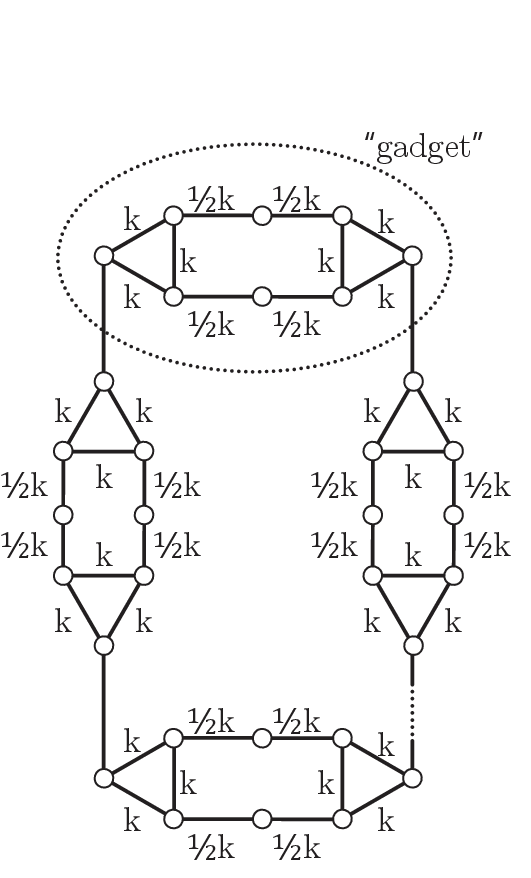}
		\caption{The edge costs.}\label{65WorseCaseExample}
	\end{subfigure}
	\hfill
	\begin{subfigure}[t]{0.45\textwidth}
		\vskip 0pt
		\centering
		\includegraphics{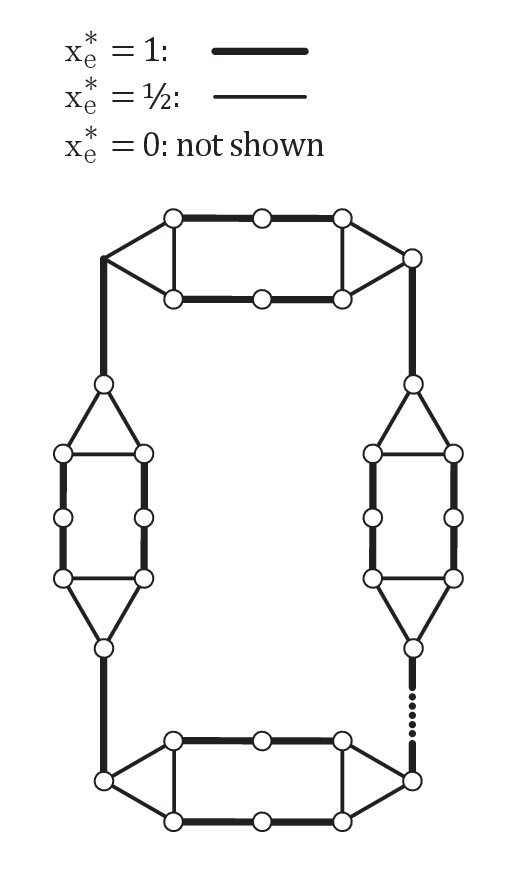}
		\caption{The half-triangle optimal solution $x^{*}$.}\label{65WorseCaseExample2}
	\end{subfigure}
	\caption{An example for which $\alphaTwoEC=\frac{6}{5}$\cite{alexander}.}
\end{figure}

The main result of this paper is to show that Conjecture\,\ref{conjecture} is true for any cost function optimized at half-triangles solutions. More specifically, we show that for any half-triangle solution $x^{*}$ and any cost function $c\ge 0$, there exists a solution of \TwoEC{} of cost at most $\frac{6}{5}cx^{*}$, which implies that $\alphaTwoEC=\frac{6}{5}$ for any cost function optimized at half-triangle solutions. Note that previously, $\frac{4}{3}$ was known, as Carr and Ravi\,\cite{carr} showed that for any degree-tight half-integer solution $x^{*}$ and any cost function $c\geq 0$, there exists a solution of \TwoEC{} of cost at most $\frac{4}{3}cx^{*}$.

A key idea used in our methods is that of convex combination. In the context of this paper, given a graph $G=(V, E)$, we say that a vector $y\in\Real^{E}$ is a \emph{convex combination} if there exist 2-edge connected spanning multi-subgraphs $H_{i}$ with multipliers $\lambda_{i}\in\Real_{\geq 0}, i=1, 2, \dotsc, j$ such that $y=\sum_{i=1}^{j}\lambda_{i}\chi^{E(H_{i})}$ and $\sum_{i=1}^{j}\lambda_{i} = 1$. Here $\chi^{E(H_{i})}\in\Real^{E}$ is the \emph{incidence vector} of subgraph $H_{i}$ (i.e. $\chi^{E(H_{i})}_{e}$ is the number of copies of edge $e$ in $H_{i}$). Our method is essentially an averaging argument, and can be described as follows: let $x^{*}$ be any feasible solution of \TwoECLP, and suppose we can show that $kx^{*}$ is greater than or equal to a convex combination for some value $k$ (in particular $k=\frac{6}{5}$). Then for any non-negative cost vector $c$ we have $kcx^{*}\ge \sum_{i=1}^{j} \lambda_{i} c\chi^{E(H_{i})}$. This implies that for at least one of the $H_{i}$, $c\chi^{E(H_{i})}\le kcx^{*}$. If $c$ is optimized at $x^{*}$ for \TwoECLP, we then have $cx^{*}=\text{OPT}(\TwoECLP)$, $\text{OPT}(\TwoEC)\leq c\chi^{E(H_{i})}$, and thus $\frac{\text{OPT}(\TwoEC)}{\text{OPT}(\TwoECLP)}\le k$ for $c$.

\section{Main Result}

Given a graph $G=(V, E)$, we sometimes use $E(G)$ to denote $E$, and $V(G)$ to denote $V$. A graph $G$ is called \emph{cubic} if every vertex of $G$ has degree three. A \emph{cut} in $G$ is a set of edges whose removal disconnects $G$ into two components, sometimes referred to as the \emph{shores} of the cut. We call a cut \emph{proper} if both shores have cardinality at least two. Given a vector $y$ that is a convex combination, the \emph{occurrence} of an edge $e$ in that convex combination is $y_{e}=\sum(\lambda_{i}:e\in H_{i})$. We sometimes refer to the occurrence of a pattern \emph{A} of edges in a convex combination, in which case we mean $\sum(\lambda_{i}:\text{pattern \emph{A} occurs in } H_{i})$, and we use the notation $\lambda_{A}$ to denote it.

In this section, we prove our main result which is that $\frac{6}{5}x^{*}$ can be expressed as a convex combination for any half-triangle solution $x^{*}$. We do this by first considering the cubic graph we get by shrinking all half-triangles to pseudo-vertices and replacing all 1-paths by singles edges. We obtain a convex combination result for this cubic graph, then show how we can use this result and certain patterns for the half-triangle edges to obtain the result that $\frac{6}{5}x^{*}$ is a convex combination.

\begin{definition}
$P(G) \Leftrightarrow $ Given a cubic 3-edge connected graph $G=(V, E)$, the vector $y^{*}\in\Real^{E}$ defined by $y^{*}_{e} = \frac{4}{5}$, for all $e \in E$, is a convex combination in which none of the 2-edge connected spanning subgraphs use more than one copy of any edge in $E$.
\end{definition}

\begin{lemma}\label{thmThreeSeventh}
$P(G)$ holds for all cubic 3-edge connected graphs $G=(V, E)$ with $|V| \geq 4$.
\end{lemma}

\begin{proof}
Suppose the contrary, and let $G=(V, E)$ be the smallest counter-example for which $P(G)$ does not hold. Since $P(G)$ can be shown to be true directly for the unique graph $G$ with $|V| = 4$ (see Figure \ref{BaseCase}, where bold lines indicate edges in $H_{i}$ and dotted lines indicate edges omitted), we can assume $|V| > 4$.

\begin{figure}
\begin{center}
\includegraphics{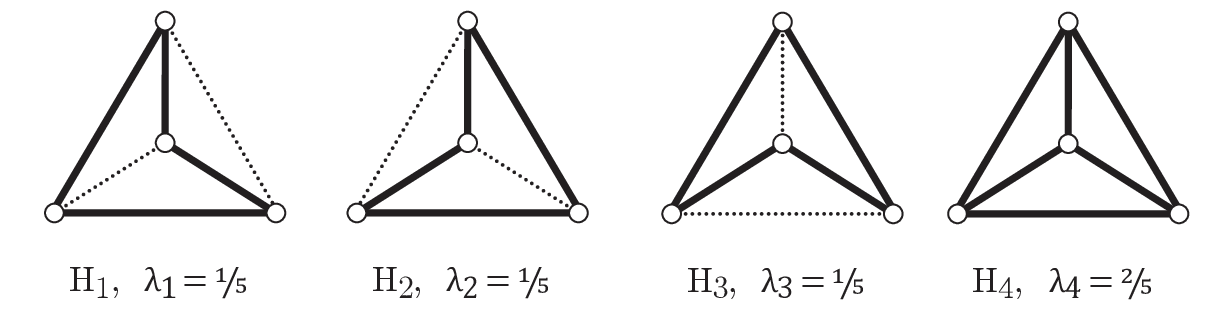}
\end{center}
\caption{Proof of Lemma \ref{thmThreeSeventh} for $G=(V, E)$, when $|V|=4$.}
\label{BaseCase}
\end{figure}
\vspace{1em}
\begin{case} \label{LemmaBaseCase}
$G$ has no proper 3-edge cut.
\end{case}
For any edge $uv\in E$, let the unlabeled adjacent vertices at $u$ be $a$ and $b$, and the unlabeled adjacent vertices at $v$ be $c$ and $d$. Since $G$ is 3-edge connected, has no proper 3-edge cut and $|V| > 4$, it follows that $a$, $b$, $c$ and $d$ are all distinct. This situation is illustrated on the left of Figure\,\ref{NoNonTrivial3EC}, where some incident edges are not shown for vertices $a$, $b$, $c$ and $d$. Removing $u$ and $v$ and their incident edges, and adding edges $ab$ and $cd$ yield a new cubic 3-edge connected graph $G'=(V', E')$ with fewer vertices than $G$. Therefore, $P(G')$ holds, so there exists a set of \mbox{2-edge} connected spanning subgraphs $H_{i}$ with multipliers $\lambda_{i}$, $i=1, 2, \dotsc, k$ such that $ab$ and $cd$ occur $\frac{4}{5}$ times overall in the convex combination. There are four patterns possible depending on the absence of $ab$ and $cd$ in $H_{i}$. These are indicated as patterns \emph{A}, \emph{B}, \emph{C} and \emph{D} in Figure\,\ref{PatternsToMissingEdge}, where an edge marked in bold indicates an edge which is in $H_{i}$, and a dotted edge indicates an edge which is not in $H_{i}$. For each pattern $Z$, we let $\lambda_{Z}$ represent the total occurrence of pattern $Z$ over all $H_{i}$ in the convex combination, i.e. $\lambda_{Z} = \sum(\lambda_{i}: \text{pattern }Z\text{ occurs in }H_{i})$.

\begin{figure}
\begin{center}
\includegraphics{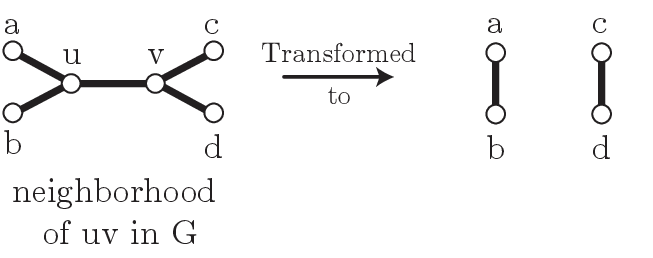}
\end{center}
\caption{Inductive step for Case \ref{LemmaBaseCase}.}
\label{NoNonTrivial3EC}
\end{figure}

Using the fact that $ab$ occurs exactly $\frac{4}{5}$ of the time, $cd$ occurs exactly $\frac{4}{5}$ of the time and $\lambda_{A} + \lambda_{B} + \lambda_{C} + \lambda_{D} = 1$, it follows that
\begin{equation} \label{lambdasEquiv}
\lambda_{A} + \lambda_{C} = \frac{4}{5} \text{, }
\lambda_{A} + \lambda_{B} = \frac{4}{5} \text{, }
\lambda_{B} + \lambda_{D} = \frac{1}{5} \text{ and }
\lambda_{C} + \lambda_{D} = \frac{1}{5} \text{.}
\end{equation}

To create a convex combination of subgraphs for $G$, we create one or two \mbox{2-edge} connected spanning subgraphs for each subgraph $H_{i}$ in the convex combination for $G'$, as shown in Figure \ref{PatternsToMissingEdge}. In the case we use two, we use multiplier $\frac{\lambda_{i}}{2}$ for each, otherwise we use multiplier $\lambda_{i}$. In Figure \ref{PatternsToMissingEdge} the resulting occurrences of the corresponding patterns in $G$ are indicated. Moreover, using (\ref{lambdasEquiv}) we have the occurrence of edges $au$ and $bu$ is $\frac{1}{2}(\lambda_{B} + \lambda_{D}) + \lambda_{A} + \lambda_{C} = \frac{9}{10}$, the occurrence of $vc$ and $vd$ is $\frac{1}{2}(\lambda_{C} + \lambda_{D}) + \lambda_{A} + \lambda_{B} = \frac{9}{10}$, and the occurrence of edge $uv$ is $\lambda_{B} +\lambda_{C}+\lambda_{D} = \frac{2}{5} - \lambda_{D}\le \frac{2}{5}$, and all the other edges occur $\frac{4}{5}$ of the time (illustrated on the right of Figure \ref{PatternsToMissingEdge}). For simplicity, we will always work with exact fractions: should the occurrence of $uv$ be less than $\frac{2}{5}$ of the time, we will add it back to arbitrary subgraphs so that it appears exactly $\frac{2}{5}$ overall.

\begin{figure}
\begin{center}
\includegraphics{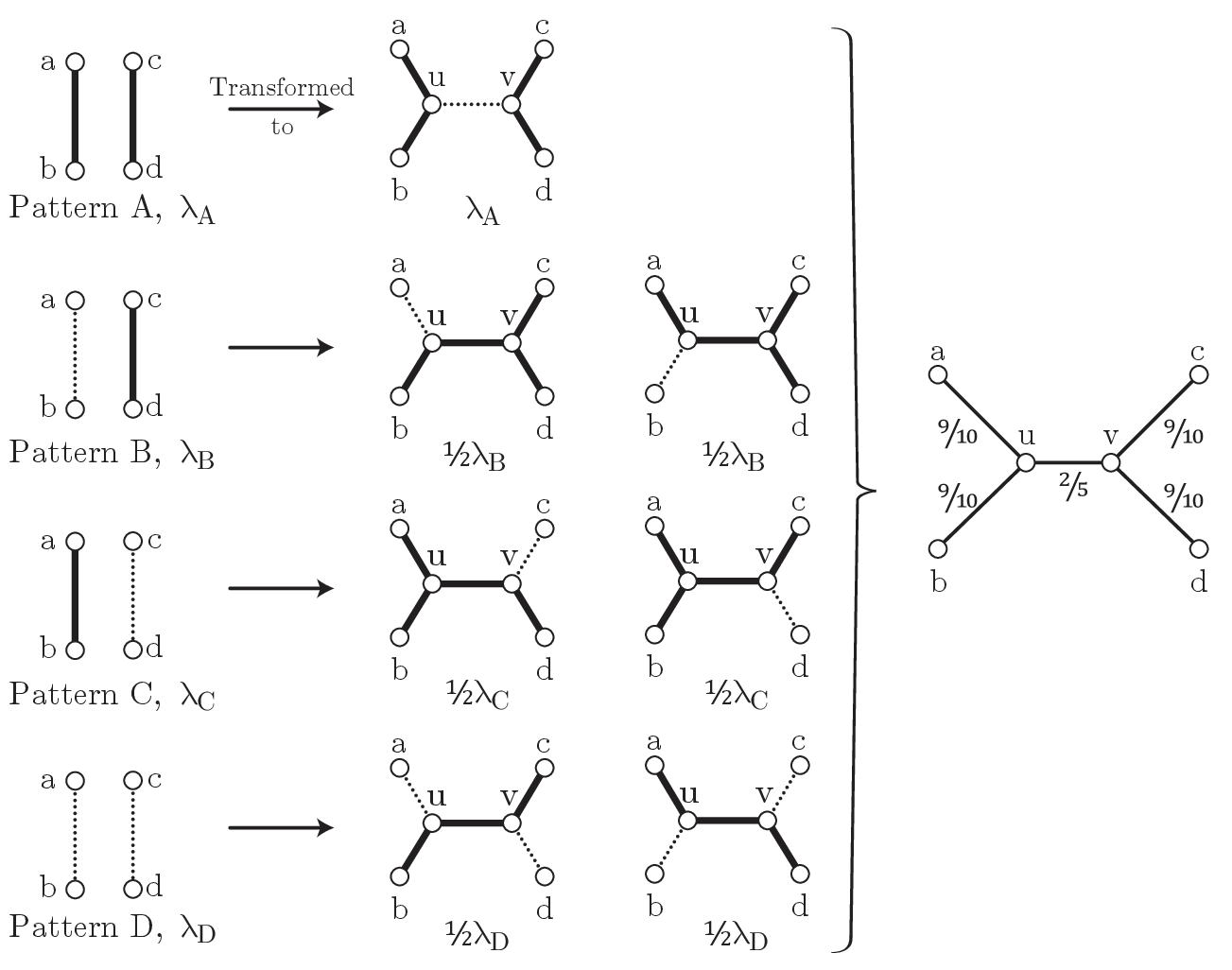}
\end{center}
\caption{Patterns for $ab$ and $cd$, and their transformations.}
\label{PatternsToMissingEdge}
\end{figure}

Applying the same technique for all edges $e \in E$ taken as edge $uv$ means that we have $m=|E|$ convex combinations, which we will refer to as $\mathbb{M}_{e}$ for each $e\in E$. Note that for any edge $f\in E$, $f$ occurs $\frac{2}{5}$ in $\mathbb{M}_{f}$, $f$ occurs $\frac{9}{10}$ in $\mathbb{M}_{e}$ for each of the four edges $e$ adjacent to $f$, and $f$ occurs $\frac{4}{5}$ in the rest of the convex combinations $\mathbb{M}_{e}$. We now take a convex combination of the $m$ convex combinations $\mathbb{M}_{e}$, $e\in E$, by multiplying every multiplier $\lambda_{i}$ used in these convex combinations by $\frac{1}{m}$. Summing the occurrence of every edge in this new convex combination gives \[\frac{1}{m}(\frac{2}{5} + \frac{9}{10}(4) + \frac{4}{5}(m-5)) = \frac{4}{5}\text{.}\]Therefore, we have a convex combination for $y^{*}$ for $G$ and $P(G)$ holds true, contradiction.

\begin{case} \label{caseTwo}
$G$ has a proper 3-edge cut.
\end{case}
Notice that the ends of the three edges must be distinct since $G$ is 3-edge connected. In this case we contract each shore of the cut to a single vertex, to obtain graphs $G_{1}=(V_{1}, E_{1})$ with pseudo-vertex $v_{1}$ and $G_{2}=(V_{2}, E_{2})$ with pseudo-vertex $v_{2}$ (as shown in Figure\,\ref{3ECSplitting}). Both $G_{1}$ and $G_{2}$ are smaller than $G$, $|V_{1}|\ge 4$ and $|V_{2}|\ge 4$, so $P(G_{1})$ and $P(G_{2})$ hold. Moreover the patterns formed by the occurrence of the edges incident to $v_{1}$ and $v_{2}$ are unique and identical in the subgraphs in the corresponding convex combinations. For instance, exactly $\frac{1}{5}$ of the time, one of the incident edges will not be in the subgraph, on both sides of the cut, and this is true for each of the three incident edges. The remaining subgraphs contain all three incident edges. These constant patterns allow us to ``glue'' (reconnect the edges as there were before the inductive step) the subgraphs for $G_{1}$ and $G_{2}$ together, in such a way that identical patterns at $v_{1}$ and $v_{2}$ are matched. This results in a convex combination for $y^{*}$ that shows $P(G)$ holds, which gives a contradiction.
\end{proof}

\begin{figure}
\begin{center}
\includegraphics{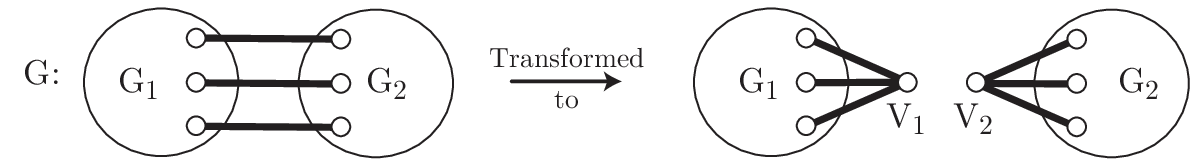}
\end{center}
\caption{Contracting both sides of a proper 3-edge cut of $G$.}
\label{3ECSplitting}
\end{figure}

We now use Lemma \ref{thmThreeSeventh} to obtain our main result below. We call a graph \mbox{$G=(V, E)$} a \emph{half-triangle graph} if $G$ is the support graph of a half-triangle solution $x^{*}$. If all 1-paths in $G$ consist of a single edge, we call $G$ \emph{simple}.

\begin{definition}\label{defThm}
$Q(G, p) \Leftrightarrow $ Given a simple half-triangle graph $G=(V, E)$ and a specified 1-edge $p\in E$, the vector $z^{*}\in \Real^{E}$ defined by
\[
 z^{*}_{e} =
  \begin{cases}
   \frac{3}{5} & \text{if } e \text{ is a half-edge of } G \text{,} \\
   \frac{4}{5} & \text{if } e = p \text{,} \\
   \frac{6}{5} & \text{otherwise,}
  \end{cases}
\]is a convex combination in which none of the 2-edge connected spanning multi-subgraphs use more than one copy of a half-edge or the edge $p$, and all of them use either one or two copies of a 1-edge.
\end{definition}

\begin{theorem}\label{mainThm}
$Q(G, p)$ holds for all simple half-triangle graphs $G=(V, E)$ and any 1-edge $p\in E$ not in a 2-edge cut in $G$.
\end{theorem}

\begin{proof}
\setcounter {case} {0}

\begin{case}\label{case1}
$G$ has no 2-edge cut.
\end{case}
\begin{figure}
\begin{center}
\includegraphics{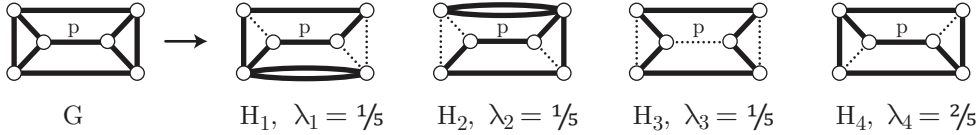}
\end{center}
\caption{Convex combination for $Q(G, p)$ when $G$ has only two triangles.}
\label{SpecialCase}
\end{figure}
If $G$ has only two half-triangles, then $Q(G, p)$ can be shown directly, using the $H_{i}$ and $\lambda_{i}$ shown in Figure \ref{SpecialCase}, where edges represented by dotted lines are omitted, and $H_{1}$ and $H_{2}$ contain a multi-edge. Otherwise, let \mbox{$G'=(V', E')$} be the graph obtained from $G$ by shrinking each half-triangle to a pseudo-vertex. Graph $G'$ is cubic and 3-edge connected and has $|V'|\geq 4$, therefore by Lemma \ref{thmThreeSeventh}, $P(G')$ holds, and yields a convex combination for $G'$ with an edge occurrence of $\frac{4}{5}$ for all edges. Let the subgraphs in this convex combination be $H_{i}', i=1, 2, \ldots, k$ with multipliers $\lambda_{1}', \lambda_{2}', \ldots, \lambda_{k}'$.

For each subgraph $H_{i}'$ in the convex combination for $G'$, the half-triangles (previously contracted to pseudo-vertices) will now be expanded to conclude the proof. We will add 1-edges and half-edges to each expanded $H_{i}'$ in such a way that we create a convex combination for the original half-triangle graph $G$ that gives the required occurrence for each edge for the theorem. To accomplish this, for each triangle $T$ in each subgraph, we will add half-edges in patterns, where each pattern is used a fraction of the time (either $\frac{1}{2}$, or $\frac{1}{3}$). To facilitate this, we simply assume that we start with a new convex combination of $G'$ which contains six copies of each subgraph $H_{i}'$, where each copy has a coefficient of $\frac{\lambda_{i}'}{6}$. 

Now consider any triangle $T$ in $G$ and let its incident edges be $x$, $y$ and $z$. In the convex combination created for $G'$, we have all three of these edges, or just two of these edges occur in each subgraph $H_{i}'$. Let 
\begin{equation*}
\begin{split}
&\lambda_{xyz}'=\sum(\lambda_{i}':\{x, y, z\}\in H_{i}') \text{,} \\
&\lambda_{xy}'=\sum(\lambda_{i}':\{x, y\}\in H_{i}') \text{,} \\
&\lambda_{yz}'=\sum(\lambda_{i}':\{y, z\}\in H_{i}') \text{ and} \\
&\lambda_{xz}'=\sum(\lambda_{i}':\{x, z\}\in H_{i}') \text{.}
\end{split}
\end{equation*}
Note that $\lambda_{xyz}' + \lambda_{xy}' + \lambda_{yz}' + \lambda_{xz}' = 1$, and each of the 1-edges $x$, $y$ and $z$ occur $\frac{4}{5}$ of the time, thus each is missing exactly $\frac{1}{5}$ of the time. Thus
\begin{equation} \label{newLambdas}
\lambda_{xyz}'=\frac{2}{5} \text{ and }\lambda_{xy}' = \lambda_{yz}' = \lambda_{xz}' = \frac{1}{5}\text{.}
\end{equation}

First we consider any expanded triangle $T$ which is not incident with edge $p$. For each subgraph $H_{i}'$ in which all three edges $x$, $y$ and $z$ occur, we include two of the three edges in $T$ $\frac{1}{3}\lambda_{xyz}'$ of the time. These patterns and their corresponding occurrences are illustrated in Figure\,\ref{ExpandingPseudoVerticesNoEdgeMissing}, and result in an occurrence of $\frac{2}{3}\lambda_{xyz}'=\frac{4}{15}$ for each edge of $T$ overall, by (\ref{newLambdas}). Note that using each pattern one third of the time can be accomplished by using the patterns of Figure\,\ref{ExpandingPseudoVerticesNoEdgeMissing} for $T$ for two of the six copies of each $H_{i}'$ where $x$, $y$ and $z$ occur. Then for each subgraph $H_{i}'$ in which $z$ is omitted and $x$ and $y$ occur, we consider both triangle $T$ and  the other triangle $T'$ incident with $z$. In this case we include the edges in $T$ incident with $z$ $\frac{1}{2}\lambda_{xy}'$ of the time, and the other edge in $T$ $\frac{1}{2}\lambda_{xy}'$ of the time, and do the opposite in triangle $T'$. In all cases we also include two copies of edge $z$. The patterns are illustrated in Figure\,\ref{ExpandingPseudoVerticesOneEdgeMissing} and result in an occurrence of $\frac{1}{2}\lambda_{xy}'$ for each edge in $T$. Note that using each pattern half of the time can be accomplished by using each of the two patterns shown in Figure\,\ref{ExpandingPseudoVerticesOneEdgeMissing} for $T$ (and $T'$) in three of the six copies of each $H_{i}'$ in which $z$ is omitted. We do the same for the cases where $x$ or $y$ are omitted in $H_{i}'$. The total occurrence of each half-edge in $T$ is 
\begin{equation*}
\begin{split}
&\frac{2}{3}\lambda_{xyz}' + \frac{1}{2}\lambda_{xy}' + \frac{1}{2}\lambda_{yz}' + \frac{1}{2}\lambda_{xz}' \text{,}
\end{split}
\end{equation*}
which by (\ref{newLambdas}) is $\frac{17}{30} < \frac{3}{5}$. We can arbitrarily add back half-edges in the convex combinations to obtain an occurrence of exactly $\frac{3}{5}$ for these edges (for a complete illustration of the operations and the pattern occurrences, see Figure\,\ref{ExpandingPseudoVertices}).

\begin{figure}
\begin{center}
\includegraphics{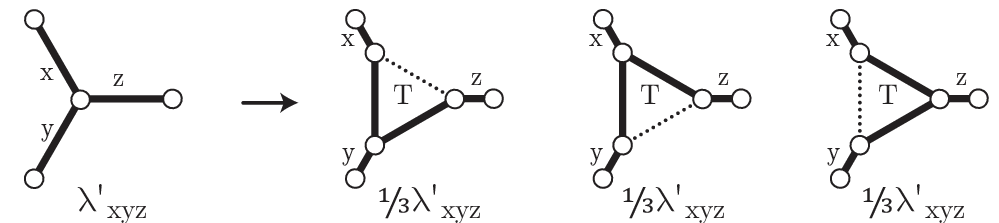}
\end{center}
\caption{Patterns used for triangle expansion for subgraphs containing $x$, $y$ and $z$.}
\label{ExpandingPseudoVerticesNoEdgeMissing}
\end{figure}

\begin{figure}
\begin{center}
\includegraphics{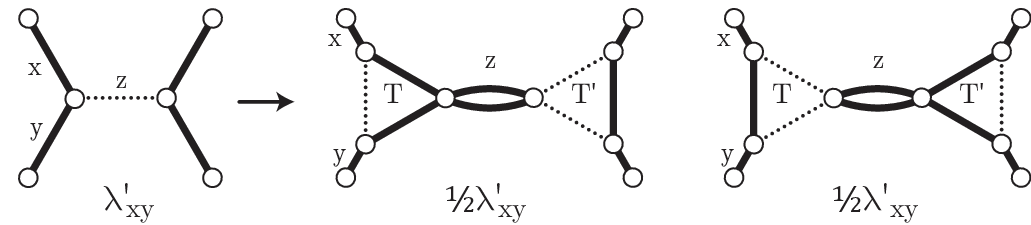}
\end{center}
\caption{Patterns used for triangle expansion for an omitted edge $z$.}
\label{ExpandingPseudoVerticesOneEdgeMissing}
\end{figure}

Note that each 1-edge which is not $p$ is now doubled whenever it was previously omitted, and thus occurs $\frac{6}{5}$ of the time. Also note that all patterns used in the expansion of the half-triangles ensure that the new multi-subgraphs created from the subgraphs $H_{i}'$ for $G'$ are also 2-edge connected and spanning in $G$, as required. 

\begin{figure}
\begin{center}
\includegraphics{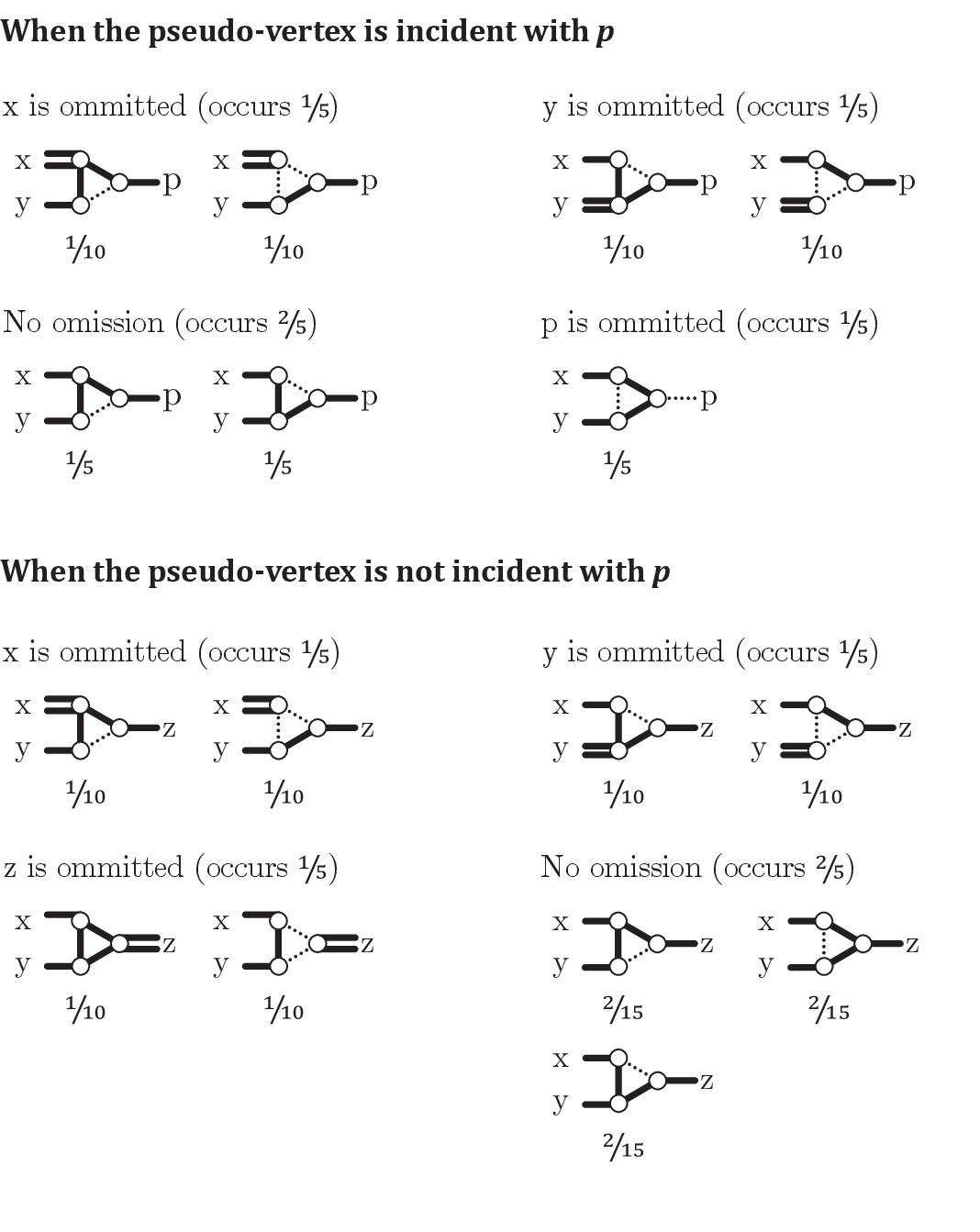}
\end{center}
\caption{Examples of edge selection upon expanding the pseudo-vertices of $G'$.}
\label{ExpandingPseudoVertices}
\end{figure}

Next we consider any expanded triangle $T$ which is incident with edge $p$, and WLOG let $p=z$. For each subgraph $H_{i}$ in which all three edges $x$, $y$ and $p$  occur, we include two of the three edges in $T$ $\frac{1}{2}\lambda_{xyp}'=\frac{1}{5}$ of the time, using the two patterns illustrated in Figure\,\ref{ExpandingPseudoVerticesPNotMissing} (so we use each pattern in three of the six copies of $H_{i}'$). Then for each subgraph $H_{i}'$ in which $p$ is omitted and $x$ and $y$ occur we include the two edges of $T$ incident with $p$ in any of these operations. Recall this occurs $\lambda_{xy}'=\frac{1}{5}$ of the time, by (\ref{newLambdas}). Note that we do not double edge $p$. The total occurrence of each edge of $T$ is exactly $\frac{3}{5}$, and $p$ occurs exactly $\frac{4}{5}$ of the time (see Figure \ref{ExpandingPseudoVertices} for a complete illustration of these operations and the pattern occurrences).

\begin{figure}
\begin{center}
\includegraphics{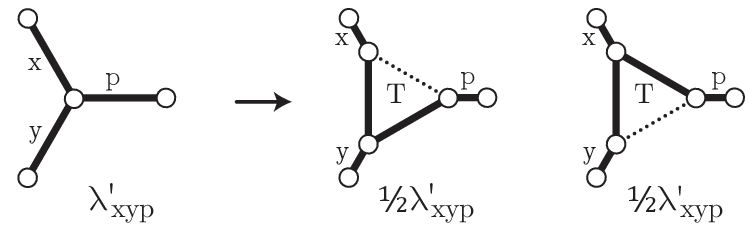}
\end{center}
\caption{Patterns used for triangle expansion for subgraphs containing $x$, $y$ and $p$.}
\label{ExpandingPseudoVerticesPNotMissing}
\end{figure}

We now have, over all cases, the half-edge occurrence is $\frac{3}{5}$, $p$ occurs $\frac{4}{5}$ of the time, and the occurrence of the other 1-edges is $\frac{6}{5}$. Furthermore, none of the 2-edge connected spanning multi-subgraphs use more than one copy of a half-edge or the edge $p$, and all of them use either one or two copies of a 1-edge. Thus $Q(G, p)$ holds.

\begin{case}
$G$ has a 2-edge cut $C=\{hi, jk\}$.
\end{case}
Suppose the contrary, and let $G$ be the smallest counter-example for which $Q(G, p)$ does not hold. Let $G_{1}$, $G_{2}$ be the two sides of the cut $C$ in $G$, with $h$ and $j$ in $G_{1}$ and $i$ and $k$ in $G_{2}$, and WLOG choose $C$ such that $G_{1}+hj$ is 3-edge connected and does not contain $p$. By smaller example and Case\,\ref{case1}, $Q(G_{1}+hj, hj)$ and $Q(G_{2}+ik, p)$ hold. We now ``glue'' together in the obvious way, the subgraphs in the convex combination for $G_{1}+hj$ where $hj$ is omitted with the subgraphs in the convex combination for $G_{2}+ik$ which have $ik$ doubled (both patterns occur\,$\frac{1}{5}$ of the time) by removing the double edge $ik$ and adding two copies of edges $hi$ and $jk$. Similarly, we glue the subgraphs for $G_{1}+hj$ and $G_{2}+ik$ where $hj$ and $ik$ occur as single edges in the subgraphs (both patterns occur $\frac{4}{5}$ of the time) by removing $hi$ and $ik$ and adding edges $hi$ and $jk$. We obtain $Q(G, p)$, contradiction.
\end{proof}

By replacing 1-edges by 1-paths in the convex combinations for $Q(G, p)$, and doubling the path for $p$ wherever $p$ was omitted, we can obtain $\frac{6}{5}x^{*}$ as a convex combination for any half-triangle solution $x^{*}$, i.e. there exist \mbox{2-edge} connected spanning multi-subgraphs $H_{i}$ with multipliers $\lambda_{i}\in \Real_{\geq 0}$, $i=1, 2, \dotsc, j$ such that $\sum_{i=1}^{j}\lambda_{i}=1$ and 
\begin{equation} \label{costEqn}
\frac{6}{5}x^{*} = \sum_{i=1}^{j}\lambda_{i}\chi^{E(H_{i})} \text{.}
\end{equation}
Now consider any non-negative cost vector $c\in \Real^{E}$ which is optimized at $x^{*}$ for \TwoECLP{}, i.e. $cx^{*} = \text{OPT}(\TwoECLP{})$. By multiplying both sides of (\ref{costEqn}) by $c$, we obtain \[\frac{6}{5}\text{OPT}(\TwoECLP)=\sum_{i=1}^{j}\lambda_{i}c\chi^{E(H_{i})}\] and thus, for at least one subgraph $H_{i}$ in the convex combination,
\begin{equation} \label{resultEqn}
c\chi^{E(H_{i})} \leq\frac{6}{5}\text{OPT}(\TwoECLP{}) \text{.}
\end{equation}
Since $\text{OPT}(\TwoEC{})\leq c\chi^{E(H_{i})}$ and $cx^{*}=\text{OPT}(\TwoECLP)$, it follows that $\frac{\text{OPT}(\TwoEC)}{\text{OPT}(\TwoECLP)}\le \frac{6}{5}$ for such cost functions. As there exists a family of half-triangle solutions which show $\alphaTwoEC \geq \frac{6}{5}$ asymptotically\,\cite{alexander}, we obtain the following corollary to Theorem \ref{mainThm}.

\begin{corollary}
The integrality gap $\alphaTwoEC=\frac{6}{5}$ when restricted to cost functions optimized at half-triangle solutions.
\end{corollary}

\end{document}